\documentclass[a4paper,12pt]{article}
\usepackage[utf8]{inputenc}
\usepackage[english]{babel}
\usepackage{amsmath}
\usepackage{amsthm}
\usepackage{latexsym}
\usepackage{amssymb}
\usepackage{amscd}
\usepackage{color}
\usepackage{graphicx}
\usepackage{graphics}
\usepackage{psfrag}
\usepackage{anysize}
\usepackage{epstopdf}
\usepackage{epsfig}
\usepackage{pst-all}
\usepackage{pstricks}
\usepackage{url}
\usepackage{float}

\marginsize{3cm}{3cm}{2cm}{2cm}
\DeclareGraphicsExtensions{.JPEG,.bmp,.pdf,.eps}
\DeclareFixedFont{\fiverm}{OT1}{cmr}{m}{n}{5pt}

\theoremstyle{definition}
\newtheorem{Definition}{Definition}

\newtheorem{remark}{Remark}

\newtheorem{Proposition}{Proposition}

\begin{document}

\title{A model of anaerobic digestion for biogas production using Abel equations}
\author{Primitivo B. Acosta-Hum\'anez\footnote{Facultad de Ciencias B\'asicas y Biom\'edicas, Universidad Sim\'on Bol\'{\i}var, Barranquilla, Colombia} \\ Maximiliano Machado Higuera\footnote{Faculty of Natural Sciences and Mathematics,  University of Ibagu\'e, Ibagu\'e, Colombia.} \\Alexander V. Sinitsyn\footnote{Department of Mathematics, National University of Colombia, Bogot\'a D.C., Colombia.}}

\date{}
\maketitle

{\bf\large Abstract:}
Some time ago has been studied mathematical models for biogas production due to its importance in the use of control and optimization of re\-new\-able resources and clean energy. In this paper we combine two algebraic methods to obtain solutions of Abel equation of first kind that arise from a mathematical model to biogas production formulated in France on 2001. The aim of this paper is obtain Liouvillian solutions of Abel's equations through Hamiltonian Algebrization. As an illustration, we present graphics of solutions for Abel equations and solutions for algebrized Abel equations.
\bigskip

$2010$ Mathematics Subject Classification: $12H05$ $34A05$, $34A34$, $90B30$.

%\linenumbers
%%%%%%%%%%%%%%%%%%%%%
\section{Introduction}\label{Section1}
%%%%%%%%%%%%%%%%%%%%%

In this paper we present some complementary results presented in \cite{mmavs2015,mmh2015t}, where was studied the following system:
\begin{equation}
\begin{array}{lll}
\frac{dX_{1}}{dt} &=& \left( \mu_{1}(S_1) - \alpha D \right) \ X_{1}\\
\frac{dS_{1}}{dt} &=& D \left( S_{1in} - S_1 \right) - k_1 \ \mu_{1}(S_1) \ X_{1}
\end{array}\label{sys1}
\end{equation}
which is related with anaerobic digestion model to biogas.
\bigskip

The following concepts and definitions can be found in  \cite{acsbook,aclm,acmowe}, although they are very well known in the literature concerning to \emph{differential Galois theory}.

A differential field $(K,\partial)$ is a field $K$ endowed with a derivation $\partial$. The field of constants of $K$, denoted by $C_K$, corresponds to all elements of $K$ vanishing under the derivation $\partial$ and it is assumed as algebraically closed and of characteristic zero.

Consider a differential equation \begin{equation}\label{liouv}f(x,y,y',y'',\ldots y^{(n)})=0,\end{equation} where $f$ is an analytic function in a complex variable $x$ and its coefficients belong to the differential field $K$. We say that a function $\theta$ satisfying equation \eqref{liouv} is either
\begin{enumerate}
	\item \emph{algebraic} over $K$ whether $\theta$ satisfies a polynomial equation whose coefficients belong to $K$, or
	\item  \emph{primitive} over $K$ whether $\theta'\in K$, i.e., $\theta=\int \eta$ for some $\eta\in K$, or
	\item \emph{exponential} over $K$ whether ${\theta'\over \theta}\in K$, i.e., $\theta=\exp(\eta)$ for some $\eta\in K$.
\end{enumerate}
A solution of equation \eqref{liouv} is called \emph{Liouvillian} whether $\theta$ is there exists a finite chain of differential fields $K=K_0\subset K_1\subset \ldots K_n$ such that $K_i=K_{i-1}(\theta)$, $1\leq i\leq n$, where $\theta$ is algebraic, primitive or exponential over $K$.

In other words, Liouvillian functions include Elementary functions such as algebraic, logarithmic, trigonometric, hyperbolic, inverses, etc. In the usual terminology of differential equations, by analytical or exact solutions we means solutions that can be obtained explicitly, including special functions. For instance, there are analytic or exact solutions that are not Liouvillian, for example Airy functions. Therefore \emph{error function} is of Liouvillian type. Thus, we consider in this paper only Liouvillian functions as solutions of non-linear equations. In \cite{acsu} were used Liouvillian functions to obtain explicit \emph{propagators}, while in this paper we obtain explicit solutions of Abel equations that are Liouvillian functions.

Our results are concerning to Abel equation of first kind, which is a nonlinear differential equation (NDE) that is cubic in unknown functions and it is written in general form as:
\begin{equation}
y^{'} = f_{0}(x) + f_{1}(x)y + f_{2}(x)y^{2} + f_{3}(x)y^{3}\label{ecuaSal}
\end{equation}
The difficult question is searching Liouvillian solutions to Abel's differential equation of first kind \eqref{ecuaSal}. The general Liouvillian solution to Abel's equation \eqref{ecuaSal} is an open problem. There are not many results on the construction of Liouvillian solutions of Abel's equation \eqref{ecuaSal} in evident form. We mention the work presented in \cite{Salinas} where is suggested a new construction method of Liouvillian solutions of \eqref{ecuaSal}. In \cite{almp} were considered Abel's equation of second kind.
\bigskip

The aim of this paper is to install coupling between the biogas subsystem \eqref{sys1}, the method to obtain Liouvillian solutions of Abel's equations presented in \cite{Salinas} and Hamiltonian Algebrization method presented in \cite{acsbook,acs} to investigate Liouvillian solutions of Abel's equation \eqref{ecuaSal}.
\bigskip

The paper contains the  following. In Section 2, following \cite{Salinas} we obtain Abel equation from NDEs modeling of upper and lower solutions to  biogas production with special coefficients $f_{0}(x), f_{1}(x), f_{2}(x), f_{3}(x)$ proposed in \cite{mmh2015t}. In the first step, we consider the nonlinear system of ODEs of biogas production, which includes seven equations as in \cite{Ber01}. In the second step following \cite{mmavs2015,mmh2015t} we decompose original biogas system into subsystems using the method of upper and lower solutions to reduce the first subsystem for two equations into the Abel's equation. In Section 3, following \cite{acs,acmowe}, we use Hamiltonian Algebrization to transform Abel's equation with transcendental coefficients in a new Abel's equation with rational coefficients and then comparing numerical approximations of Abel's equation solutions with numerical approximations of Algebrization of Abel's equation solutions. Finally in the Section 4 we compute the Liouvillian solutions of the algebrized Abel's equation.
\bigskip

%%%%%%%%%%%%%%%%%%%%%%%%%%%%%%%%%%%%
\section{Biogas subsystem associated with Abel equation}
%%%%%%%%%%%%%%%%%%%%%%%%%%%%%%%%%%%%

Biogas - generic name for a combustible gas mixture produced in
the decomposition of organic substances due to anaerobic
microbiological process (methane fermentation). Biogas, the
end-product of anaerobic digestion (AD) mainly consisting of
$CH_{4} (60-70\%)$ and $CO_{2} (30-40\%)$, provides considerable
potential as a versatile carrier of renewable energy, not solely
because of the wide range of substrates that can be used for the
AD process \cite{Sialve}. AD has been used to convert biomass into biogas microalgae \cite{Sialve,Mariet}. AD is also used in wastewater treatment \cite{Mariet,Ber01}. Biogas is widely used as a combustible fuel en Germany, Denmark, China, United States and other developed countries. It is used for home consumption and pu\-blic transportation.
\vspace{0.1cm}

We consider the following nonlinear system of ODEs that models the process of AD \cite{Ber01}:
\vspace{0.1cm}
{\bf{\it{Biomass balance}}}
\begin{align}
\frac{dX_{1}}{dt} &= \left( \mu_{1}(S_1) - \alpha D \right) \ X_{1} \ \ \	 \text{(acidogenesis)}\label{ecua1}\\
\frac{dX_{2}}{dt} &= \left( \mu_{2}(S_2) - \alpha D \right) \ X_{2} \ \ \ \text{(methanogenesis)}\label{ecua2}
\end{align}

{\bf{\it{Substrate balance}}}
\begin{align}
\frac{dS_{1}}{dt} &= D \left( S_{1in} - S_1 \right) - k_1 \ \mu_{1}(S_1) \ X_{1}\label{ecua3}\\
\frac{dS_{2}}{dt} &= D \left( S_{2in} - S_2 \right) + k_2 \ \mu_{1}(S_1) \ X_{1} - k_3 \ \mu_{2}(S_2) \ X_{2}\label{ecua4}
\end{align}

{\bf{\it{Alkalinity balance}}}
\begin{align}
\frac{dA}{dt} & = D \left( A_{in} - A \right)\label{ecua5}
\end{align}

{\bf{\it{Carbon exchange rate}}}
\begin{align}
\frac{dC}{dt} & = D \left( C_{in} - C \right) + k_4 \ \mu_{1}(S_1) \ X_{1} + k_5 \ \mu_{2}(S_2) \ X_{2}\\
\;\; & \; - K_{L_a}[C + S_2 - A - K_{H}P_{C}]\label{ecua6}
\end{align}
The product  $K_H \ P_C = B$ determines concentration of oxygen dissolved in $C$.
\vspace{0.1cm}

{\bf{\it{Net rate of methane production}}}
\begin{align}
\frac{d F_M}{dt} & = k_6 \ \mu_{2}(S_2) \ X_2\label{ecua7}
\end{align}

with methane concentration \textbf{$F_M$}.
\vspace{0.1cm}

Nonlinear kinetic behavior, occurs due to the reaction rates, which are given by  $\mu_{1}(S_{1}) = \mu_{1max}\frac{S_{1}}{S_{1} + K_{S_{1}}}$ - Monod type kinetic and $\mu_{2}(S_{2}) = \mu_{2max} \frac{S_{2}}{\frac{S_{2}^{2}}{K_{I2}} + S_{2} + K_{S_{2}}}$ - Haldane type kinetic represent bacterial growth rates associated with two bioprocesses.
\vspace{0.1cm}

In this case the variables are:
\begin{align*}
S_{1}:=& \ {\rm Organic\; substrate\; concentration \ [g/l]}\\
X_{1}:=& \ {\rm Concentration\; of\; acetogenic\; bacteria \ [g/l]}\\
S_{2}:=& \ {\rm Volatile\; fatty\; acids\; concentration \ [mmol/l]}\\
X_{2}:=& \ {\rm Concentration\; of\; methanogenic\; bacteria \ [g/l]}\\
A:=& \ {\rm Concentration\; of\; alkalinity \ [mmol/l]}\\
C:=& \ {\rm Total\; inorganic\; carbon\; concentration \ [mmol/l]}\\
F_{M}:=& \ {\rm Methane\; concentration \ [mmol/l \;\; d^{-1}]}.
\end{align*}

With $\mu_{1max}>0$, $\mu_{2max}>0$, $K_{S_{1}}>0$, $K_{S_{2}}>0$, $K_{I2}>0$,
$A_{in}>0$, $B>0$, $C_{in}>0$, $K_{L_{a}}>0$, $S_{1in}>0$,
$S_{2in}>0$, and $k_{i}>0$, $i=1,2,3,4,5,6$., are positive
constants. The parameter $\alpha$ ($0\leq \alpha \leq 1$)
therefore reflects this process heterogeneity: $\alpha =0$
corresponds to an ideal fixed-bed reactor, whereas $\alpha =1$
corresponds to and ideal continuously stirred tank reactor(CSTR) \cite{Jerome,Ber01}.
\vspace{0.1cm}

\begin{Definition}
A trivial solution of the system \eqref{ecua1}-\eqref{ecua7} has the form
\begin{equation*}
\begin{aligned}
&E_{1}(0,S_{1in}), \;\; E_{2}( 0,S_{1in},0,S_{2in}), \;\; E_{3}(0, S_{2in})\\
&\;\;\;\;\;\;\; E_{4}( 0,S_{1in},0,S_{2in},A_{in}), \;\; E_{5}(A_{in}) \label{soltriv}
\end{aligned}
\end{equation*} \label{defin1}
where $X_{1}=0, \;X_{2}=0, \;S_{1}=S_{1in}, \;S_{2}=S_{2in}, \;A=A_{in}, \;C=C_{in}, \;F_{M} = 0$.
\end{Definition}

We are interested in the solution of the nonlinear subsystem with initial conditions, which models dynamics of biomass and substrate in acidogenesis.

\begin{equation}
\left\{
\small{
\begin{aligned}
&\frac{dX_{1}}{dt}  = \left( \mu_{1_{max}} \frac{S_1}{S_1 + K_{S_1}} - \alpha D \right) \ X_{1}  \triangleq F\left( t,X_{1}(t),S_{1}(t) \right),\\
&\frac{dS_{1}}{dt}  = D \left( S_{1in} - S_1 \right) - k_1 \ \mu_{1_{max}} \frac{S_1}{S_1 + K_{S_1}} \ X_{1} \triangleq  G(t,X_{1}(t),S_{1}(t)),\\
& X_{1}(0) \; \; = \; \; c_{1},\\
& S_{1}(0) \; \; = \; \; c_{2}.
\end{aligned} }
\right. \label{sis3}
\end{equation}
where $ t\in [0,T]=I,$ with $T>0 $ and $F, G$ are functions of $C^{0}(I)=C(I)$ class.
\vspace{0.1cm}

\begin{Definition}[Lower-upper solution]
A pair 
{\small{ $[(X_{10}, S_{10})$, $(X_{1}^{0}, S_{1}^{0})]$ }}
is called
\begin{description}
\item[] Lower-upper solution of the problem \eqref{sis3}, if the following conditions are sa\-tis\-fied
    \begin{align*}
%     \begin{array}
      &(X_{10}, S_{10}) \in C^{1}(I), \qquad ( X_{1}^{0}, S_{1}^{0} ) \in C^{1}(I), \qquad t \in I\\
      &\dot{X_{1}}_{0} - F(t, X_{10}, S_1) \leq 0  \qquad {\rm (lower)}\\
      &\dot{X_{1}}^{0} - F(t, X_{1}^{0}, S_1) \geq 0 \ {\rm \; in}\;\; \; I, \;\;\; \forall S_1 \in [S_{1}^{0}, S_{10}], {\rm \;\; (upper)}\\
      &\dot{S}_{10} - G(t, X_{1}, S_{10}) \leq 0 \qquad  {\rm \; (reverse \; order)}\\
      &\dot{S}^{0}_{1}- G(t, X_{1}, S^{0}_{1}) \geq 0\ {\rm \; in}\; I, \;\forall X_{1}\in [X_{10}, X_{1}^{0}]  {\rm  (reverse \; order)}\\
      &{\rm and\; with\; initial\; conditions}\\
      &X_{10}(0) \leq c_2 \leq X_{1}^{0}(0), \; S_{1}^{0}(0) \leq c_1 \leq S_{10}(0);
%     \end{array}
    \end{align*}
with $X_{10}\leq X^{0}_{1}, \;\;\; S^{0}_{1}\leq S_{10}   \;\;{\rm \; in}\;\;I.$
\end{description}\label{defin2}
\end{Definition}

We reduce the system \eqref{sis3} via transformation  $S_{1}=S_{1in}$ to the one nonlinear equation
$$\frac{dX_{1}}{dt}  = \left( \mu_{1_{max}} \frac{S_1}{S_1 + K_{S_1}} - \alpha D \right) \ X_{1}$$
with solution
\begin{align*}
X_{1} &= X_{1}(0)\exp{\left(\mu_{1max}\frac{S_{1in}}{S_{1in} + K_{S_1}} - \alpha D\right)t}\notag\\
X_{1} &= X_{1}(0)\exp{\left(  m\;t \right)}\;\;\; {\textrm{with}}\;\;\; m=\frac{\mu_{1max} S_{1in}}{S_{1in} + K_{S_1}} - \alpha D.\label{ecua19}
\end{align*}
We take $\Gamma>0$ therefore, the solution
\begin{equation*}
X_{1}^{0} = X_{1} + \Gamma \;\;\; {\rm and}\;\;\; X_{10} = X_{1} - \Gamma \label{lowX1triv}
\end{equation*}
are an upper solution and lower solutions of \eqref{sis3} respectively.

From Definition \ref{defin2} to ${S_1}$ in \eqref{sis3}
\begin{equation}
\begin{aligned}
{\dot{S_{1}}}_{0} & - D \left( S_{1in} - {S_1}_{0} \right) + k_1 \; \mu_{1max} \frac{{S_1}_{0}}{{S_1}_{0} + K_{S_1}} \; X_{1}^{0} \leq 0 \ \ \ \textrm{(lower)}\\
{\dot{S_{1}}}^{0} & - D \left( S_{1in} - {S_1}^{0} \right) + k_1 \; \mu_{1max} \frac{{S_1}^{0}}{{S_1}^{0} + K_{S_1}} \; X_{10} \geq 0 \; \; \; \textrm{(upper)}
\end{aligned}
\end{equation}

We analyze the case where $X_{10} \neq 0$, on \eqref{sis3} to search upper-solution for $S_1$, then
\begin{equation*}
\frac{dS_{1}^{0}}{dt} = D (S_{1in} - {S_1^0}) - \left[ \frac{k_1 \ \mu_{1max} \ S_1^0}{S_1^0 + K_{S_1}} \right] X_{10}\notag
\end{equation*}
taking $X_{1}^{0}$  from \eqref{lowX1triv}:

\begin{equation}
\frac{dS_{1}^{0}}{dt} = D (S_{1in} - {S_1^{0}}) - \left[ \frac{k_1 \ \mu_{1max} \ S_1^{0}}{S_1^{0} + K_{S_1}} \right]
\left[ X_{1}(0)\ \exp{\left( m\;t \right) - \Gamma} \right]\label{upS1}
\end{equation}

we have $V=(S_1^{0} + K_{S_1})^{-1}$ is obtained Abel equation the first kind
{\small{
\begin{equation}
\begin{aligned}
\frac{dV}{dt} &= DV \ + \ \left[- D( S_{in} + K_{S_1} ) + k_1 \mu_{1max}X_{1}(0) \exp{\left( mt \right) } - k_{1}\mu_{1max}\Gamma \right] V^2\\
& + \left[ - k_1 \ \mu_{1max} \ X_{1}(0) \ K_{S_1}  \exp{\left( m\;t \right) + k_{1}\mu_{1max}\Gamma } \right] V^3 .
\end{aligned}\label{ecuaAbel}
\end{equation} } }

%%%%%%%%%%%%%%%%%%%%%%%%%%%%%%%%%%%%%%%%%%%%%%%
\section{Algebrization of Abel equations}\label{Section2}
%%%%%%%%%%%%%%%%%%%%%%%%%%%%%%%%%%%%%%%%%%%%%%%

In this section we transform the equation \eqref{ecuaAbel}, Abel equation, into a new Abel equation with rational coefficients. To do this, we use the \emph{Hamiltonian Algebrization}. For suitability, from now on we write $\partial_t$ instead of $\frac{d}{dt}$.

\begin{Definition}[Hamiltonian Change of Variable]
The change of variable $w=w(t)$ is called Hamiltonian whether there exists $\alpha=\alpha(w)$ such that $(\partial_t w)^2=\alpha(w)$. If $\alpha(w)\in \mathbb{C}(w)$, we say that $w=w(t)$ is a rational Hamiltonian change of variable.
\end{Definition}
We call \emph{Algebrization of a differential equation} whether we can obtain an algebraic form of such differential equation. This means that we transformed a differential equation with non-rational coefficients into a differential equation with rational coefficient. The Algebrization is called Hamiltonian whether was done through a Hamiltonian change of variable. Following \cite{acmowe} we obtain in a natural way the following result.\\
\begin{Proposition}
The algebraic form of Abel equation \eqref{ecuaSal}, through the Hamiltonian change of variable $w=\exp(mt)$, is given by
{\small{
\begin{equation}
\begin{aligned}
\partial_w\widehat{V} &= \frac{D}{mw}\widehat{V} \ + \left[{- D( S_{in} + K_{S_1} ) + k_1 \mu_{1max}X_{1}(0) w - k_{1}\mu_{1max}\Gamma\over mw} \right]\widehat{V}^2\\
& + \left[ {- k_1 \ \mu_{1max} \ X_{1}(0) \ K_{S_1} w + k_{1}\mu_{1max}\Gamma  \over m w}\right]\widehat{V}^3 ,
\end{aligned}\label{ecuaAbelalg}
\end{equation} } }
\end{Proposition}

\begin{proof}
Due to the change of variable $w=\exp(mt)$ is Hamiltonian, we can obtain $\alpha=\alpha(w)=(\partial_t w)^2$. In this case $\alpha=m^2 w^2\in\mathbb{C}(w)$. Furthermore,  $\partial_ t$ is transformed in $\sqrt{\alpha}\partial_w=m w\partial_w$, which allow us to arrive to our result. That is,
from the change of variable
$$
w = e^{mt},   
$$
with
$$
\frac{1}{mw} \frac{dV}{dt} = \partial _{w} \hat{V}
$$

we obtain
\begin{equation*}
\begin{aligned}
\partial _{w}\widehat{V} &= \frac{D}{m w}\widehat{V} \ + \left[{- D( S_{in} + K_{S_1} ) + k_1 \mu_{1max}X_{1}(0) w - k_{1}\mu_{1max}\Gamma\over m w} \right]\widehat{V}^2\\
& + \left[ {- k_1 \ \mu_{1max} \ X_{1}(0) \ K_{S_1} w + k_{1}\mu_{1max}\Gamma  \over m w}\right]\widehat{V}^3.
\end{aligned}
\end{equation*}

\end{proof}

\begin{remark}
Algebrized Abel equation \eqref{ecuaAbelalg} can lead us to the obtaining of solutions of Abel equation \eqref{ecuaSal}. Moreover, the hamiltonian change of variable $w=e^{mt}$ and its derivative is a solution curve of the hamiltonian $H(w,p)={p^2\over 2}-{m^2w^2\over 2}+const.$
\end{remark}

We rewrite \eqref{ecuaAbelalg} to the form
\begin{equation}
\partial_{w}\hat{V} = f_{1} \hat{V} + f_{2}\hat{V}^{2} + f_{3}\hat{V}^{3}\notag\label{ecuaAbelalg2}
\end{equation}
where
$$
f_{1}= \frac{D}{mw}, \;\; f_{2}=\frac{a_{5}+a_{2}w}{mw}, \;\; f_{3}=\frac{a_{4}w-a_{3}}{mw}
$$
\begin{equation}
a_{1} = -D(S_{1in}+K_{S_1}), \;\; a_{2} = k_{1}\mu_{1max}X_{1}(0), \;\; a_{3}=- k_{1}\mu_{1max} \Gamma, \;\; a_{4}=-K_{S_1}a_{2}, \;\; a_{5}=a_{1}+a_{3}.\label{const}
\end{equation}

%%%%%%%%%%%%%%%%%%%%%%%
\subsection*{Numerical approximations}
%%%%%%%%%%%%%%%%%%%%%%%

First we reduce Abel equation with change of variable
\begin{equation}
V=\frac{1}{S_{1}^{0}+K_{S_{1}}}\label{varialV}
\end{equation}
to plotting $S_{1}(t)$.
\smallskip

\begin{figure}[!hbt]			% Grafica de  S1	via Abel
\centering
\includegraphics[scale=0.33]{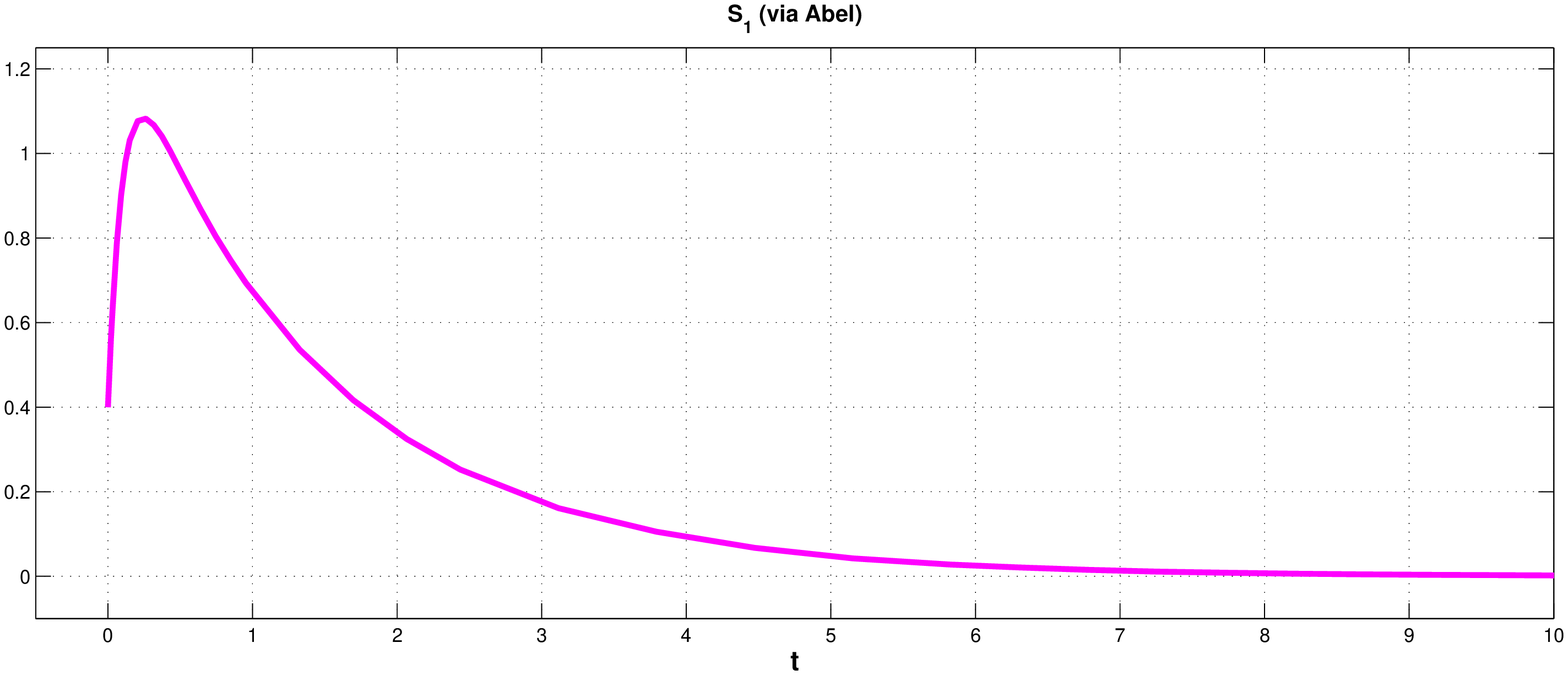}
\caption{$S_1^0$. }
\label{figura1}
\end{figure}
Figure \ref{figura1} shows  the graph of the numerical solution for the varibale $S_{1}$ associated with the Abel's equation.
\smallskip

\begin{table}[!hbt]
\begin{center}
%\begin{minipage}{\textwidth}
\scalebox{0.85}{
\begin{tabular}{||l|c|l|c||} \hline
\textbf{Parameter} & \textbf{Value} & \textbf{Units} & \textbf{SD}\\ \hline
$\alpha$ & 0.5 & \;\; & 0.4 \\ \hline
$D$ & 0.395 & [d$^{-1}$] & 0.135 \\  \hline
$S_{1in}$ & 10 & [g/l] & 6.4 \\ \hline
$K_{S_1}$ & 12.1 & [g/l] & 20.62 \\ \hline
$\mu_{1max}$ & 1.2 & [d$^{-1}$] & \;\; \\  \hline
$k_{1}$ & 23.2 & \;\; & \;\; \\ \hline
\end{tabular}
}
%\end{minipage}
\caption{Parameter for simulation (from \cite{Ber01}.) SD = standard deviation}
\label{tabla1}
\end{center}
\end{table}
\smallskip

Then we take Algebrization of Abel equations with change of variable
$$
\hat{V}(w) =  \frac{1}{S_{1}^{0} + K_{S_{1}}} 
$$
with
$$
w = e^{mt},   \;\;\;\;\;\;\;\; w(0) = 1
$$
and
$$
t = \frac{1}{m} \ln{w}
$$
to plotting $S_{1}^{0}(t)$\\
\begin{figure}[!hbt]						% Grafica de  S1	via Abel Algebrizada
\centering
\includegraphics[scale=0.33]{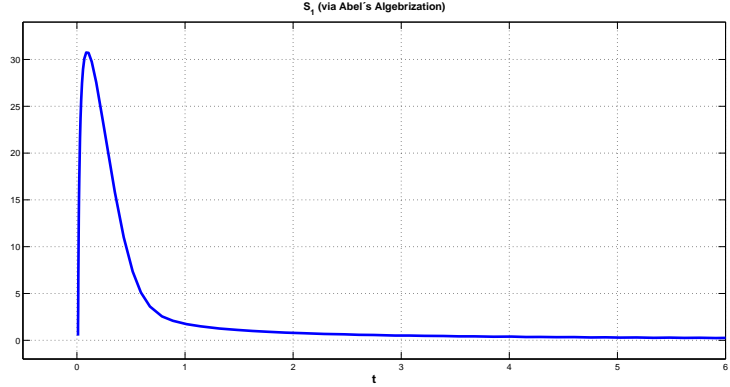}
\caption{$S_1^0$ via algebrized Abel.}
\label{figura2}
\end{figure}
\newline

Figure \ref{figura2} shows  the graph of Liouvillian solution for the variable $S_{1}$ associated with the algebrized Abel's equation.
\bigskip
The Parameter for simulation for the Figure \ref{figura2} from Table \ref{tabla1}.
\bigskip

%%%%%%%%%%%%%%%%%%%%%%%%%%%%%%%%%%%%
\section{Exact solutions of Abel equation}
%%%%%%%%%%%%%%%%%%%%%%%%%%%%%%%%%%%%

We apply the method presented \cite{Salinas} to the construction of Liouvillian solutions for Abel equation in the algebraic form \eqref{ecuaAbel}. The transformation $y=u(x)z(x)+v(x)$ reduces original Abel's equation to the canonic form 
$$
z'-z^{3}-\Phi =0
$$
with
\begin{equation}
\begin{aligned}
\Phi (\xi) &= \frac{\varphi e^{\int{\varphi}d\xi}}{\sqrt{C-\int{e^{2{\int{\varphi}d\xi}}}d\xi}}\notag\\
 &= \frac{1}{f_{3}u^{3}} \left[  f_{0} - \frac{f_{1}f_{2}}{3f_{3}} + \frac{2f_{2}^{3}}{27f_{3}^{2}} + \frac{1}{3}\frac{d}{dx} \frac{f_{2}}{f_{3}} \right],\notag\\
u(x)&= e^{\int{\left( f_{1}-\frac{f_{2}^{2}}{3f_{3}} \right) }dx}, \; \xi = \int{f_{3}u^{2}}dx,\notag
\end{aligned}
\end{equation}
where $\varphi$ is arbitrary function and $v=-\frac{f_{2}}{3f_{3}}$.\\

{\bf Case 1}. \; $\varphi(\xi) =0$. In this case we obtain the solution of Abel equation \eqref{ecuaAbelalg} in the following form
\begin{equation}
y= \frac{e^{\int{\left( f_{1}-\frac{f_{2}^{2}}{3f_{3}} \right) }dx}}{\sqrt{C-\int{f_{3}e^{2{\int{\left( f_{1}-\frac{f_{2}^{2}}{3f_{3}} \right)}dx}}}dx}} - \frac{f_{2}}{3f_{3}}\label{solSal}
\end{equation}

{\bf Case 2}. \; $\varphi(\xi) = C_{0}$. where $C_{0}$ is a constant.

\begin{equation}
y=e^{\int{\left( f_{1}-\frac{f_{2}^{2}}{3f_{3}} \right) }dx} \frac{e^{C_{0}\int{f_{3}e^{2\int {\left( f_{1}-\frac{f_{2}^{2}}{3f_{3}} \right)dx }}dx}}}{\sqrt{C-\int e^{2C_{0}\int{f_{3}e^{2\int{\left( f_{1}-\frac{f_{2}^{2}}{3f_{3}} \right)}dx}}}}dx} - \frac{f_{2}}{3f_{3}}\label{solSal2}
\end{equation}

To obtain Liouvillian solutions, we consider the phenomenon of whashing, when the variable $X_{1}\longrightarrow 0$. In this case $X_{1}(0)=0$ and we reduce solution \eqref{solSal} in the form 

\begin{equation}
\hat{V}= \frac{e^{F_{1}(w)}}{\sqrt{C-F_{2}(w)}} - \frac{f_{2}}{3f_{3}}\notag
\end{equation}
with
\begin{equation}
\begin{aligned}
F_{1}(w)&= \int{\left(  f_{1}-\frac{f_{2}^{2}}{3f_{3}} \right) dw} = \ln{\left| C_{1}w^{b_{1}} \right|}; \;\;\; F_{2}=2F_{1}(w)\notag\\
F_{3}(w)&= \int{f_{3}e^{2F_{1}(w)}}dw = \frac{a_{3}C_{1}^{2}}{2b_{1}m} w^{2b_{1}} + C_{2}\notag
\end{aligned}
\end{equation}
and
$$
b_{1}=\frac{D}{m}-\frac{a_{5}^{2}}{3ma_{3}}, \;\;\; b_{2}=\frac{C_{0}a_{3}C_{1}^{2}}{b_{1}m} 
$$
with constants, as defined in \eqref{const}. Applying the change of variable $\hat{V}[e^{mt}]= V(t)$, the Liouvillian solution to Abel equation Case 1, is
$$
V(t)=\frac{C_{1}e^{b_{1}mt}}{\sqrt{C-2(\ln{C_{1}}+b_{1}mt)}} - \frac{a_{5}}{3a_{3}}
$$
Taking up again the change of variable \eqref{varialV}, the solutions to \eqref{upS1} is
$$
S_{1}^{0} = \frac{3a_{3}\sqrt{C-2(\ln{C_{1}}+b_{1}mt)}}{3a_{3}C_{1}e^{b_{1}mt}-a_{5}\sqrt{C-2(\ln{C_{1}}+b_{1}mt)}} - K_{S-1}
$$

\begin{figure}[!hbt]						% Grafica de  S1	via Abel Algebrizada
\centering
\includegraphics[scale=0.33]{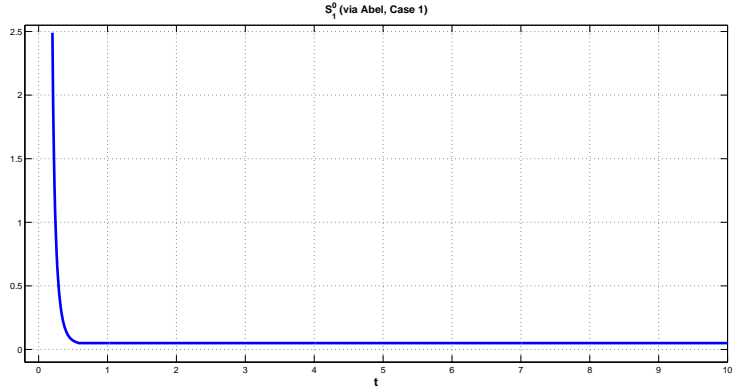}
\caption{$S_1^0$ via Abel, Case 1. }
\label{figura3}
\end{figure}
Figure \ref{figura3} shows the graph of Liouvillian solution for the variable $S_{1}$ associated with the Abel's equation to Case 1.
\smallskip

Now reduce solution \eqref{solSal2} in the form

\begin{equation}
\hat{V}=e^{F_{1}(w)} \frac{e^{c_{0}} F_{3}(w)}{\sqrt{C-F_{4}(w)}} - \frac{f_{2}}{3f_{3}}\notag
\end{equation}
with
$$
F_{4}(w)=\int{e^{2C_{0}F_{3}(w)}dw} = -\frac{e^{2C_{0}C_{2}}}{2b_{1}(-b_{2})^{1/2b_{1}}} \; \Gamma{\left( \frac{1}{2b_{1}},-b_{2}w^{2b_{1}} \right)} + C_{3}
$$

Using the change of variable above, the solution Abel equation to Case 2 is
$$
V(t)=C_{1}e^{b_{1}mt}\; \frac{ \frac{C_{0}a_{3}C_{1}^{2}}{2b_{1}m} e^{2b_{1}mt}+C_{2}}{\sqrt{C-C_{3} + \frac{e^{2C_{0}C_{2}}}{2b_{1}(-b_{2})^{1/2b_{1}}} \Gamma{\left( \frac{1}{2b_{1}},-b_{2}e^{2mb_{1}} \right)} }} - \frac{a_{5}}{3a_{3}}.
$$

From \eqref{varialV}, the solution to $S_{1}^{0}$ is

$$
S_{1}^{0} = \frac{3a_{3} \sqrt{C-C_{3} + \frac{e^{2C_{0}C_{2}}}{2b_{1}(-b_{2})^{1/2b_{1}}}\;\Gamma{\left( \frac{1}{2b_{1}}, -b_{2}e^{2b_{1}mt} \right)} } }{3a_{3}C_{1}e^{b_{1}mt}\left( \frac{C_{0}a_{3}C_{1}^{2}}{2b_{1}m} e^{2b_{1}mt} + C_{2}\right) -a_{5}\sqrt{C-C_{3} + \frac{e^{2C_{0}C_{2}}}{2b_{1}(-b_{2})^{1/2b_{1}}}\;\Gamma{\left( \frac{1}{2b_{1}}, -b_{2}e^{2b_{1}mt} \right)} } }.
$$

%%%%%%%%%%%%%%%%%%%%%%%
\section*{Final Remarks}
%%%%%%%%%%%%%%%%%%%%%%%

In this paper we found a direct connection between the variable that represents the substrate in the Acidogenesis process and Abel's equation of first kind. In the search for Liouvillian solutions to this equation, we implemented a combination between Hamiltonian Algebrization with a method to obtain Liouvillian solutions of Abel's equations given in \cite{Salinas}. We found Liouvillian solutions for case 1 and case 2. The graphics show that the solution of  differential equation for $S_{1}$ has the same behavior as the graphics of the numerical approximations, which tells us that these functions represent the process found.
\bigskip

{\bf Acknowledgements.} The first author thanks to Universidad de Ibagu\'e for the support to finish this work through a research stay there. The second author acknowledges to Universidad de Ibagu\'e for supporting his research through project ref: 15-344-INT.

\end{document}